\documentclass{article}
\usepackage{microtype}
\usepackage{graphicx}
\usepackage{subfigure}
\usepackage{booktabs} 
\usepackage{tikz}
\usepackage{enumitem}
\usepackage{wrapfig}

\usepackage{hyperref}

\usepackage[accepted]{icml2019}

\icmltitlerunning{Hiring Under Uncertainty}

\usepackage{amsmath,amsthm,amsfonts,amssymb,xcolor}
\usepackage{tikz}

\def\algo{\mathsf{alg}_1}
\def\algt{\mathsf{alg}_2}
\def\opt{\mathsf{opt}}
\def\greedy{\mathsf{greedy}}
\def\val{\mathsf{val}}
\def\seg{\mathsf{segval}}

\def\seqalg{\mathsf{seqalg}}
\def\seqopt{\mathsf{seqopt}}

\def\paralg{\mathsf{paralg}}
\def\parheur{\mathsf{parheur}}
\def\paropt{\mathsf{paropt}}
\def\size{\mathsf{size}}
\def\tree{\mathcal{T}}
\def\candidates{\mathcal{C}}

\def\dx{\; dx}

\newcommand{\p}[1]{\left( #1 \right)}
\renewcommand{\b}[1]{\left[ #1 \right]}
\newcommand{\E}[1]{\mathbb{E} \b{#1}}
\newcommand{\EM}[2]{\mathbb{E}_{#1} \b{#2}}
\newcommand\numberthis{\addtocounter{equation}{1}\tag{\theequation}}
\newcommand{\prob}{{\textsc{Hiring with Uncertainty}}}

\newtheorem{lemma}{Lemma}
\newtheorem{example}{Example}
\newtheorem{claim}[lemma]{Claim}
\newtheorem*{claim*}{Claim}

\makeatletter
\renewcommand*{\ALG@name}{Procedure}
\makeatother

\begin{document}

\twocolumn[
\icmltitle{Hiring Under Uncertainty}\

\icmlsetsymbol{equal}{*}

\begin{icmlauthorlist}
\icmlauthor{Manish Raghavan}{cornell}
\icmlauthor{Manish Purohit}{goo}
\icmlauthor{Sreenivas Gollapudi}{goo}
\end{icmlauthorlist}

\icmlaffiliation{cornell}{Department of Computer Science, Cornell University}
\icmlaffiliation{goo}{Google, Inc.}

\icmlcorrespondingauthor{Manish Raghavan}{manish@cs.cornell.edu}

\icmlkeywords{Machine Learning, IML}

\vskip 0.3in
]

\printAffiliationsAndNotice{}

\begin{abstract}
In this paper we introduce the hiring under uncertainty problem to model the questions faced by hiring committees in large enterprises and universities alike. Given a set of $n$ eligible candidates, the decision maker needs to choose the sequence of candidates to make offers so as to hire the $k$ best candidates. However, candidates may choose to reject an offer (for instance, due to a competing offer) and the decision maker has a time limit by which all positions must be filled. Given an estimate of the probabilities of acceptance for each candidate, the hiring under uncertainty problem is to design a strategy of making offers so that the total expected value of all candidates hired by the time limit is maximized. We provide a 2-approximation algorithm for the setting where offers must be made in sequence, an 8-approximation when offers may be made in parallel, and a 10-approximation for the more general stochastic knapsack setting with finite probes.
\end{abstract}

\section{Introduction}
\label{sec:intro}

Hiring is a core activity of any enterprise where the timely fulfillment of staffing needs is critical to its functioning. In addition to estimating the quality and suitability of a candidate, the enterprise also needs to deal with uncertainty that arises as the result of good candidates rejecting the job offer. Balancing this trade-off between hiring good quality candidates while at the same time ensuring that all staffing needs are met by a deadline is one of the most challenging aspects of hiring in practice.

A number of algorithmic questions that are inspired by hiring settings have been well-studied in literature (see Section \ref{sec:rw}) including the popular \emph{secretary problem} and its many variants. This line of work focuses on the online nature of the problem and the key question tackled is how to find a good set of candidates when the pool of future candidates is unknown. However, this line of research does not model the other source of uncertainty, i.e., the candidate itself may choose to reject the job offer (for instance, due to a better competing offer), which in turn raises the question of hiring enough candidates by the deadline.

During the hiring process, the ``quality'' (or value) of a candidate is often estimated by traditional hiring processes such as resume screening and formal interviews and even via algorithmic techniques (see Section \ref{sec:rw}). On the other hand, machine learning models can estimate the probability that a given candidate will accept a job offer based on various features such as the candidate's educational background, salary expectations, location preferences, and so on.
Considering both the value as well as offer acceptance probability of each candidate leads to a rich collection of optimization problems. In this paper, we initiate the study of the hiring under uncertainty problem that aims to tackle the inherent trade-off at the heart of the hiring process - {\em how should we make job offers under uncertainty so that all staffing needs are met by a deadline, and yet hire the best available candidates}?

Formally, we consider the following model as the basis for all the variants we present in this paper. There is a set of $n$ candidates, and we need to hire $k$ of them. We do this by making offers to candidates, which we'll also refer to more abstractly as ``probing'' a candidate. Each candidate $i$ has a known value $v_i$ and probability $p_i$ of accepting an offer, independent of all other candidates. We have a deadline of $t$ time steps, after which we can't make any further offers. It takes one time step to make an offer and receive an answer from a candidate. Job offers are irrevocable, i.e., once a candidate accepts an offer, that position is ``filled'' and we cannot replace that candidate with a better candidate in the future. The total value of a chosen set of candidates is simply the sum of the individual candidate values. Our goal is to maximize the total expected value of the hired candidates. We also consider two natural generalizations of this model. First, we allow making parallel offers to multiple candidates in a given time step. Second, we consider the knapsack hiring problem where each candidate $i$ has a size $s_i$ and we have a budget $B$ on the total size of hired candidates. The knapsack hiring problem models the scenario when the enterprise has a fixed budget and different candidates need to be offered different salaries.

We note that in all settings, we do not require $v_i$ to be known precisely; all of our results hold if $v_i$ is only known in expectation.
However, our results are sensitive to errors in $p_i$ and $s_i$.
Making them robust to such errors is an interesting subject for future work.

\subsection{Our Contributions}
We summarize our contributions in this study.
\begin{itemize}
    \vspace{-0.1in}
    \item In Section \ref{sec:star}, we offer a $2$-approximation algorithm for hiring $k$ candidates with a constraint of making at most $t$ sequential offers.
    \vspace{-0.1in}
    \item In Section \ref{sec:parallel}, we consider the parallel offers model where we are allowed to make as many parallel offers each time step as the number of unfilled positions remaining and design a  $8$-approximation algorithm. 
    \vspace{-0.1in}
    \item In Section \ref{sec:knapsack10}, we present a $10$-approximation for the {\em knapsack hiring} problem where each candidate has a different size and the decision-maker is constrained by a total budget (as opposed to hiring $k$ candidates). 
    \vspace{-0.1in}
    \item We offer a connection to other stochastic optimization problems such as stochastic matching and present a lower-bound for the stochastic matching problem.
    \vspace{-0.1in}
    \item Finally, we show the efficacy of our algorithms using simulations on data drawn from different distributions.
\end{itemize}

\subsection{Related Work}
\label{sec:rw}
Theoretical questions inspired by hiring scenarios have long been studied in the
online setting under the names of optimal stopping or ``secretary''
problems~\cite{dynkin1963optimum,chow1964optimal}. A few extensions of this
setting incorporate elements of our model. Kleinberg considers the
case of hiring multiple candidates instead of the traditional single-hire
case~\cite{kleinberg2005multiple}. An older line of work considers a version of
the secretary problem in which candidates may stochastically reject offers,
although this is typically modeled as a fixed rejection probability~\cite{smith1975secretary,tamaki1991secretary,tamaki2000minimal,ano2000note}.

In addition, more recent work on stochastic optimization considers a variety of
related problems in the offline setting. This includes stochastic versions of
submodular optimization~\cite{asadpour2008stochastic,gupta2017adaptivity},
knapsack~\cite{dean2004approximating,dean2005adaptivity,bhalgat2011improved},
bandits~\cite{gupta2011approximation,ma2014improvements}, and
matching~\cite{bansal2010lp,adamczyk2015improved,baveja2018improved}. Some
special cases of our model (specifically, when one candidate is being hired) can
be considered a special case of matching, and in fact, the results we derive
here will provide lower bounds for stochastic matching. However, our model
cannot in general be captured by any of these prior works.

Algorithmic and data-driven approaches to hiring have become increasingly common
with the rise of machine learning~\cite{miller2015can,carmichael2015hiring}. In
particular, there is a long line of work focused on predicting teacher quality
from
data~\cite{kane2008estimating,dobbie2011teacher,chalfin2016productivity,jacob2018teacher}.
More broadly,~\citet{mullainathan2017machine} describe the integration of
machine learning with traditional econometric techniques to better estimate
quantities like employee performance. Furthermore, studying the gig economy,
\citet{kokkodis2015hiring} use machine learning to estimate the likelihood that
freelancers get hired.

\section{Hiring Problem: How to fill $k$ positions sequentially?}
\label{sec:star}

In this section, we consider the basic hiring problem where we want to hire $k$ employees out of $n$ potential candidates with a constraint of making at most $t$ sequential offers.

\subsection{Special case: Hiring a single employee ($k = 1$)}
\label{sec:singlecandidate}

To develop some intuition about the problem as well as to illustrate some of the challenges posed, we begin with the case where $k=1$, i.e.\ , we only want to hire one candidate.

One might hope that a simple greedy algorithm is optimal in this special case. Unfortunately, as we will show, a number of seemingly natural greedy algorithms\footnote{For instance, sorting the candidates by decreasing $p_i$, $v_i$, or $p_i \cdot v_i$ and then making at most $t$ offers until one accepts.} do not yield optimal solutions. 

However, we can still take advantage of structural properties of the solution. In particular, given a set of $t$ candidates, the optimal order in which to make offers to them is in decreasing order of $v_i$. To see why, for any two candidates $i$ and $j$, consider the four possible outcomes of making offers to them: both $i$ and $j$ accept, both reject, $i$ accepts and $j$ rejects, and vice versa. The only outcome in which the order of offers matters is when they both accept, since the position will go to the candidate receiving first offer, and the second offer will never be made. In this case, it is clearly better to make the first offer to the candidate with higher value.

Since the optimal algorithm must always make offers to candidates in decreasing order by value, we can write a dynamic program to compute the optimal subset of $t$ candidates to potentially make offers to.
Assume the candidates are sorted in non-increasing order of $v_i$, i.e., $v_1 \geq v_2 \geq \ldots \geq v_n$. Let $S(i, s)$  be the optimal expected value that can be achieved with $s$ time steps remaining by only considering candidates $i$ through $n$. Then, we have the recurrence
\[
S(i, s) = \max\{p_i v_i + (1-p_i) S(i+1, s-1), S(i+1, s)\}
\]
where the two terms correspond to either making an offer to candidate $i$ or not. Note that $S(1, t)$ then gives the value of the optimal solution, and the offer strategy can be found by retracing the choices of the dynamic program.

\subsection{General Problem: Hiring $k$ employees ($k > 1$)}
While the $k=1$ case admits a clean solution, the general case where $k > 1$ is more complex.
We first note that a simple $k$-approximation exists: using the dynamic program from Section~\ref{sec:singlecandidate}, we know how to optimally fill a single slot. Doing so yields a candidate who is in expectation at least as good as any of the $k$ candidates hired by the optimal strategy.

In general, the optimal solution may display several non-monotonicities that make it difficult to extend the $k=1$ solution. 

\begin{example}
\label{ex:non_monotone}
Consider the following instance with $n = 4$, $t = 3$, and $k = 2$.
\begin{align*}
    (p_1, v_1) &= (1, 1) &
    (p_2, v_2) &= (0.5, 1) \\
    (p_3, v_3) &= (0.5, 1) &
    (p_4, v_4) &= (0.1, 2)
\end{align*}
\end{example}
We will show that in the optimal strategy, the offers made are not necessarily monotone in acceptance probability, value, or expected value.
First, note that any deterministic strategy can be represented as a binary decision tree, where each node in the tree corresponds to a candidate to whom an offer is made.
The two branches are the resulting strategies if the offer is accepted or rejected. Taking the convention that the right branch corresponds to acceptance, the optimal solution for the above instance is as shown in Figure~\ref{fig:ex_sol}.

\begin{wrapfigure}{l}{.19\textwidth}
\centering
    \begin{tikzpicture}[scale=.6]
    \node (2) at (0, 0) {2};
    \node (4) at (2, -1) {4};
    \node (1l) at (.8, -2) {1};
    \node (3) at (-2, -1) {3};
    \node (1m) at (-.8, -2) {1};
    \node (1r) at (-3.2, -2) {1};
    \draw[green]
    (2) -- (4)
    (3) -- (1m)
    ;
    \draw[red]
    (2) -- (3)
    (4) -- (1l)
    (3) -- (1r)
    ;
    \end{tikzpicture}
\caption{An optimal solution to Example~\ref{ex:non_monotone}}
\label{fig:ex_sol}%
\end{wrapfigure}
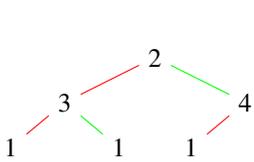

Note that there are several counter-intuitive effects at play here. First, despite having the lowest acceptance probability and expected value, candidate 4 still receives an offer with probability $1/2$. Second, the candidate with the highest expected value (candidate 1) receives an offer either last or not at all. Finally, despite the fact that candidates 1, 2, and 3 all have the same value, it is strictly optimal to make an offer to candidate 2 (or 3) before candidate 1, even though candidate 1 accepts with higher probability.

Thus, unlike in the $k=1$ scenario, the optimal solution may not be value-ordered, so the dynamic programming approach discussed above cannot be optimal here. We conjecture that this problem is NP-hard for general $k$.
In the remainder of this section, we present an approximation algorithm that runs in polynomial time and yields at least half of the optimal expected value.
We first show that there exists a non-adaptive algorithm yielding a 2-approximation.
Then, we show that a dynamic program similar to that in Section \ref{sec:singlecandidate} gives an adaptive algorithm that is better than \emph{any} non-adaptive algorithm, and hence is also a 2-approximation.

\subsubsection{Establishing an adaptivity gap of 2}
\citet{gupta2017adaptivity} study adaptivity gaps for stochastic probing problems where the goal is to maximize a given submodular function (or XOS function) over the set of active, probed elements.
In this setting, each element $e$ is active independently with probability $p_e$ and the set of elements that are probed must satisfy given prefix-closed constraints.
The \prob\ problem does not quite fit into their framework, since their framework allows one to choose the ``best'' set of active, probed elements, while in our setting we are forced to hire the \emph{first} $k$ candidates that are active.
Nevertheless, we can leverage some insights from~\cite{gupta2017adaptivity} to show an adaptivity gap of 2 (as opposed to 3 obtained by them for stochastic probing).

Similar to the one shown in Figure~\ref{fig:ex_sol}, the optimal solution to any instance can be represented by a binary tree $\tree$. Each node $u$ of $\tree$ corresponds to a candidate $i$ (denoted by $cand(u)$) and has two outgoing edges leading to subtrees in case the candidate $i$ is active (happens with probability $p_i$) or inactive (happens with probability $(1-p_i)$). Any root to leaf path in this tree represents the sequence of offers made by the optimal algorithm in a particular realization. The tree $\tree$ naturally defines a probability distribution $\pi_\tree$ over root to leaf paths - start at the root and at each node $u$, follow the ``yes'' edge with probability $p_i$ where $i = cand(u)$ and the ``no'' edge otherwise.
Since the optimal strategy can make offers to at most $t$ candidates, any such path must have at most $t$ nodes.
\begin{wrapfigure}{r}{0.17\textwidth}
\centering
    \begin{tikzpicture}[scale=.5,thick]
    \node (a) at (0, .3) {$A$};
    \node (b) at (-1, 0) {$B$};
    \node (c) at (-2, -0.5) {$C$};
    \node (d) at (-1, -1.5) {$D$};
    \node (e) at (-3, -2) {$E$};
    \node (f) at (-2, -3) {$F$};
    \node (g) at (-2.8, -3.5) {$G$};
    \node (h) at (-3.8, -4) {$H$};
    \node (i) at (-4.8, -4.5) {$I$};
    \draw[red,->] (0, 0) -> (-1, -.5);
    \draw[red,->] (-1, -.5)-> (-2, -1) ;
    \draw[green,->] (-2, -1) -> (-1.5, -1.5) ;
    \draw[red,->] (-1.5, -1.5)-> (-2.5, -2) ;
    \draw[green,->] (-2.5, -2) -> (-2, -2.5) ;
    \draw[red,->] (-2, -2.5)-> (-3, -3) ;
    \draw[red,->] (-3, -3)-> (-4, -3.5) ;
    \draw[red,->] (-4, -3.5)-> (-5, -4) ;
    \end{tikzpicture}
    \caption{A path with 3 segments: $ABC$, $DE$, and $FGHI$.}
    \label{fig:segment}%
\end{wrapfigure}
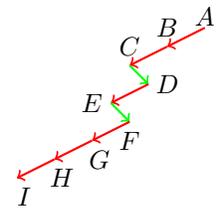
Further, since any strategy can only hire at most $k$ candidates, any root to leaf path in $\tree$ must have at most $k-1$ ``yes'' edges. Thus any root to leaf path $P$ can be decomposed into at most $k$ ``segments'' where a segment is a maximal sub-path composed of only ``no'' edges as shown in Figure~\ref{fig:segment}. Let $\mathsf{segments}(P) = \{S_1, S_2, \ldots, S_\ell\}$ denote the set of segments in $P$. For each segment $S \in \mathsf{segments}(P)$, let $\mathsf{last}(S)$ denote the \emph{last} node on segment $S$.

Given the optimal tree $\tree$, Procedure~\ref{alg:app_given_tree} samples a single path $P$ according to the distribution $\pi_\tree$ and then probes the candidates on each segment of $P$ in descending order by value to hire at most one candidate from each segment. In the rest of this section, we show that Procedure~\ref{alg:app_given_tree} yields at least half of the total expected value of $\tree$ in expectation.

\begin{algorithm}
\begin{algorithmic}[1]
\caption{ApproxGivenTree$(\tree)$}
\label{alg:app_given_tree}
\STATE $P \gets$ a random path sampled from $\pi_{\tree}$
\STATE $S_1, \dots, S_\ell \gets P$ divided into at most $k$ segments \\
\COMMENT{Each $S_j$ is a list of candidates}
\FOR {$j \gets 1, \dots, \ell$}
\STATE $S_j' \gets S_j$ sorted in decreasing order of value
\FOR{each candidate $i$ in $S_j'$}
\STATE Make an offer to candidate $i$
\IF{$i$ accepts}
\STATE \textbf{break}
\ENDIF
\ENDFOR
\ENDFOR
\end{algorithmic}
\end{algorithm}

Let $\val(\tree) = \EM{P \sim \pi_\tree}{\sum_{j = 1}^{\ell} v_{\mathsf{last}(S_j)}}$ be the total expected value of the tree $\tree$ (note that $\ell$ is a random variable). Similarly, let $\algo(\tree)$ be the expected value obtained by Procedure~\ref{alg:app_given_tree} on tree $\tree$. 
For any segment $S_j$, we define $\seg(S_j)$ to be the expected value of the active candidate from $S_j$ with largest value. Formally, if $S_j$ consists of candidates $\{1, 2, \ldots, |S_j|\}$ sorted in non-increasing order of their values, then 
\[\seg(S_j) \triangleq \sum_{i=1}^{|S_j|} p_i v_i \prod_{j < i} (1-p_j).\]
We observe that $\algo(\tree) = \EM{P \sim \pi_\tree}{\sum_{j = 1}^\ell \seg(S_j)}$. In other words, Procedure~\ref{alg:app_given_tree} obtains the value of the active element with the largest value in each segment. The following lemma shows that in expectation, this is a 2-approximation to $\val(\tree)$.

\begin{lemma}
\label{lem:seq_approx}
\[
\EM{P \sim \pi_\tree}{\sum_{j = 1}^\ell \seg(S_j)} \ge \frac{1}{2} \EM{P \sim \pi_\tree}{\sum_{j = 1}^{\ell} v_{\mathsf{last}(S_j)}}
\]
\end{lemma}

\begin{proof}
We proceed by induction over the segments.
Let $I$ be $\mathsf{last}(S_1)$,
and let $J$ be the random variable denoting the index of the first active candidate on $S_1$, so $J \le I$. If no candidate on this segment is active, we'll say $J = 0$. 

In the base case, $\ell=1$. Otherwise, using the inductive hypothesis,
\[
\EM{P \sim \pi_\tree}{\sum_{j = 2}^\ell \seg(S_j)} \ge \frac{1}{2} \EM{P \sim \pi_\tree}{\sum_{j = 2}^{\ell} v_{\mathsf{last}(S_j)}}.
\]
In either case, it suffices to show that $\E{v_J} \ge \frac{1}{2} \E{v_I}$.

This follows from Lemma 3.3 of \citet{gupta2017adaptivity} and the complete
proof is deferred to Appendix~\ref{app:proofs}.
\end{proof}

\paragraph*{Removing adaptivity.}
Note that Procedure~\ref{alg:app_given_tree} is adaptive, since it probes within a segment until it finds an active element. However, we can use it to argue about a simpler non-adaptive algorithm: pick a random path down the optimal tree $\tree$, sort all the items in it in decreasing order of value, and make offers in that order. This has value at least as large as $\algo(\tree)$ because for any realization of which elements are active and inactive, making offers in decreasing order of value is always beneficial. Thus, the adaptivity gap for this problem is at most $2$.

\subsubsection{A constructive 2-approximation}
In the above section, we have shown that there exists a non-adaptive algorithm whose total expected value is at least half of the expected value of the optimal algorithm. However, this algorithm relied on the knowledge of the optimal decision tree $\tree$ and is thus non-constructive.
We now design a polynomial time algorithm ($\seqalg$) whose expected value is at least the expected value of \emph{any non-adaptive algorithm}, and hence is also at least half the expected value of the optimal algorithm.

We observe that by definition, any non-adaptive algorithm must choose a fixed sequence of $t$ potential candidates and make offers to them in order until $k$ of them accept. Further, as discussed in Section \ref{sec:singlecandidate}, the optimal such algorithm must probe the candidates in non-increasing order by value.
However, using a dynamic programming strategy similar to that in Section \ref{sec:singlecandidate}, we can find the \emph{optimal} algorithm (not necessarily non-adaptive) that probes candidates in non-increasing order by value. This must be better than the optimal non-adaptive algorithm and hence is also a 2-approximation.

\paragraph*{Dynamic Program ($\seqalg$).}
 We again assume that the candidates are sorted in non-increasing order of their values $v_i$. Let $S(i, \ell, s)$ be the optimal expected value that can be achieved by hiring at most $\ell$ candidates in $s$ time steps by only considering candidates $i$ through $n$ in sorted order.
 We obtain the following recurrence:
\begin{align*}
S(i, \ell, s) = &\max\{p_i (v_i + S(i+1, \ell-1, s-1)) \numberthis \label{eq:mult_recurrence} \\
&+ (1-p_i) S(i+1, \ell, s-1), S(i+1, \ell, s)\}.
\end{align*}
where the two terms correspond to either making an offer to candidate $i$ or not.
Let $\seqalg$ denote the dynamic program constructed using the above recurrence. We abuse notation slightly and let $\seqalg_{k,t} = S(1, k, t)$ be the expected value obtained by this algorithm. 

Let $\seqopt_{k,t}$ be the value of the optimal adaptive strategy. Lemma~\ref{lem:seq_approx} shows that the optimal non-adaptive strategy is a 2-approximation to $\seqopt_{k,t}$. Because $\seqalg$ is at least as good as any non-adaptive strategy, we have that for a set of candidates $\candidates$, $\seqalg_{k,t}(\candidates) \ge \frac{1}{2} \seqopt_{k,t}(\candidates)$.

\paragraph*{A lower bound.}
It is an open question as to whether the above analysis is tight, i.e., whether this algorithm may actually be closer to optimal than a factor of two. 
However, 
by modifying the probabilities and values in Example~\ref{ex:non_monotone}, we
show in Appendix~\ref{app:gap} that no algorithm that provides a value-sorted solution (including $\seqalg$) can get more than $0.927$ of the optimal algorithm in general.

\section{Filling $k$ Positions in Parallel}
\label{sec:parallel}
In the previous section, we considered the problem of hiring with sequential offers. However, if we have $k$ positions to fill, we could in principle make $k$ offers per timestep. This is clearly more powerful than the sequential offer model, since any sequence of sequential offers is valid in the parallel model. We'll treat the constraint of filling $k$ positions as hard, meaning that if at a particular timestep there are $\ell < k$ remaining unfilled positions, we can only make $\ell$ offers at that time, though it would be an interesting future direction to consider a relaxed version in which we hire at most $k$ candidates with high probability.

Intuitively, the more slots remain available, the more offers can be made, which is beneficial when there are many high-value low-probability candidates. This means an optimal strategy must somehow balance the tension between two conflicting objectives: filling slots and maximizing the number of offers that can be made to risky candidates. The following example demonstrates this tension.

\begin{example}
\label{ex:parallel}
Consider the example with $n = 2t - 1$ candidates and $k=2$. $2t-2$ of the candidates have $p_i = 1/(2t-2)$ and $v_i = 1$, and the last candidate has $p_n = 1$ and $v_n = 1$.
\end{example}
Even though candidate $n$ will surely accept an offer, the optimal strategy here is to make offers to all of the low-probability candidates (2 at a time) until one of them accepts, and then to make an offer to candidate $n$, who will definitely accept.
As $t$ gets large, this yields value approximately $\frac{2e-1}{e} \approx 1.63$ in expectation,
Making an offer to candidate $n$ first can only get value approximately $\frac{2\sqrt{e}-1}{\sqrt{e}} \approx 1.39$, since we can only make one offer per timestep after we fill the first slot.
Thus, the order in which offers are made significantly impacts the overall value.

\paragraph*{An 8-approximation algorithm ($\paralg$).}
We now design $\paralg$, a constructive 8-approximation algorithm, drawing on the results in Section~\ref{sec:star}. The basic idea is to relax the parallel offer instance with $t$ timesteps to a sequential offer instance with $k \cdot t$ timesteps, solve this using $\seqalg_{k,kt}$, and use this solution to construct a solution to the original instance.

Given a set of candidates $\candidates$, let $\paropt_{k, t}(\candidates)$ be the expected value of the optimal solution with parallel offers, filling $k$ slots in $t$ timesteps. Then, $\paropt_{k, t}(\candidates) \le \seqopt_{k, kt}(\candidates)$, since any sequence of parallel offers can be done in sequence over $kt$ timesteps.

We can apply the dynamic programming algorithm $\seqalg$ from Section \ref{sec:star} to $\candidates$ to get a sequential-offer strategy over $kt$ timesteps yielding expected value at least $\frac{1}{2} \seqopt_{k, kt}(\candidates)$.
Let $\tree$ be the resulting decision tree. We'll show how to convert this sequential-offer decision tree over $kt$ timesteps into a parallel-offer strategy over $t$ timesteps.

\begin{algorithm}
\begin{algorithmic}[1]
\caption{ParallelFromSequential$(\tree)$}
\label{alg:p_from_s}
\STATE $P \gets$ a random path sampled from $\pi_\tree$
\STATE $S_1, \dots, S_\ell \gets$ the segments of $P$
\STATE $S_1', \dots, S_m' \gets$ segments split such that each has length at most $t$.
\STATE Sort each segment $S_j'$ in decreasing order of $v_i$.
\STATE Let $U$ be the indices of the $k$ segments with highest $\seg(\cdot)$
\FOR {$s \gets 1, \dots, t$}
\FOR {$j \in U$}
\STATE Make an offer to candidate $i$, the $s^{\text{th}}$ candidate from $S_j'$, at this timestep $s$.
\IF {$i$ accepts}
\STATE Remove $j$ from $U$
\ENDIF
\ENDFOR
\ENDFOR
\end{algorithmic}
\end{algorithm}

Let $\algt(\tree^*)$ be the expected value of the parallel-offer strategy produced by Procedure~\ref{alg:p_from_s}, where $\tree^*$ is the output of $\seqalg_{k,kt}(\candidates)$. Then, we have the following.
\begin{lemma}
\label{lem:par_approx}
\[
\algt(\tree^*) \ge \frac{1}{8} \paropt_{k, t}(\candidates)
\]
\end{lemma}
\begin{proof}
By Lemma~\ref{lem:seq_approx},
\[
\seqalg_{k,kt}(\candidates) \ge \frac{1}{2} \seqopt_{k, kt}(\candidates) \ge \frac{1}{2} \paropt_{k, t}(\candidates).
\]
To complete the proof, we must show that $\algt(\tree^*) \ge \frac{1}{4} \seqalg_{k,kt}(\candidates)$. Note that Procedure~\ref{alg:p_from_s} yields an offer strategy in the parallel model, while $\tree^*$ represents a strategy in the sequential model with $kt$ timesteps.

First, observe that by applying Lemma~\ref{lem:seq_approx} again, if we could make offers along each segment of a random path down $\tree^*$ in decreasing order of value, we'd get a 2-approximation to $\seqalg_{k,kt}(\candidates)$, since we'd get the maximum active element on each segment. Since there are at most $k$ segments, we could make an offer to the highest valued candidate from each segment in the first time step and proceed down each segment in parallel, discarding a segment once a candidate accepts an offer.
However, since some segments may have length more than $t$, we may not have enough offers to go all the way down each segment. Consequently, in step 3 of Procedure~\ref{alg:p_from_s}, we partition the segments further so that each new segment contains at most $t$ candidates.
More formally, if a segment $S_j \in \mathsf{segments}(P)$ has length $at + b$ for some integers $a,b \geq 0$ and $b < t$, arbitrarily split the candidates in $S_j$ into $a+1$ new segments such that $a$ of those segments have exactly $t$ candidates.

Let $|S_j|$ be the size of the $j$th segment. When we split it up into new segments of length at most $t$, $S_j$ will be turned into $\lceil \frac{|S_j|}{t} \rceil$ segments. Thus, after splitting, the number of new segments is at most
\[
\sum_{j=1}^\ell \left\lceil \frac{|S_j|}{t} \right\rceil \le
\sum_{j=1}^\ell \frac{|S_j|}{t} + 1 \le
k + \frac{\sum_{j=1}^\ell |S_j|}{t} \le 2k.
\]

Thus, if we were to pick $k$ segments uniformly at random, each segment would have probability at least $1/2$ of being selected randomly into $U$, meaning that
\begin{align*}
    \sum_{j=1}^m \E{\seg(S_j') \Pr[j \in U]}
    &\ge \frac{1}{2} \sum_{j=1}^m \E{\seg(S_j')} \\
    &\ge \frac{1}{2} \sum_{j=1}^\ell \E{\seg(S_j)}
    \\
    & \ge \frac{1}{4} \E{\seqalg_{k,kt}(\candidates)},
\end{align*}
where the last inequality follows from Lemma~\ref{lem:seq_approx}.
We can derandomize this by choosing the $k$ segments with highest expected value, which must be at least as good as $k$ random segments. As a result, we have
\[
\algt(\tree^*) \ge \frac{1}{4} \E{\seqalg_{k,kt}(\candidates)} \ge \frac{1}{8} \paropt_{k, t}(\candidates).
\]
\end{proof}

Thus, our final algorithm (which we call $\paralg_{k,t}$, or just $\paralg$ when $k$ and $t$ are clear) is to first apply $\seqalg_{k,kt}$ to $\candidates$, producing a tree $\tree^*$, and build a parallel-offer strategy from $\tree^*$ with Procedure~\ref{alg:p_from_s}.
\begin{equation}
\label{eq:paralg_def}
    \paralg_{k,t}(\candidates) \triangleq \text{ParallelFromSequential}(\seqalg_{k,kt}(\candidates))
\end{equation}
As noted above, our final approximation factor will be $\min(k, 8)$, since we get a $k$-approximation simply by filling one position optimally.

\section{Knapsack Hiring Problem}
\label{sec:knapsack10}

We now consider the knapsack hiring problem that directly generalizes the vanilla hiring problem studied in Section \ref{sec:star}. In this case, in addition to a value $v_i$ and probability $p_i$, each candidate $i$ also has a size $s_i$. Instead of a number of slots $k$, we have a budget $B$ on the total size of the hired candidates. As earlier, we have a deadline of $t$ time steps and can make only one offer per time step.

The knapsack hiring problem is closely related to the well-studied stochastic knapsack problem~\cite{dean2004approximating,dean2005adaptivity}, which is as follows: we are given $n$ items, each with a known value and ``size distribution''. When an item is added to the knapsack, its size is drawn from this distribution. Once an item exceeds the capacity, this item must be discarded and no further items can be added to the knapsack. In the multidimensional version, both the capacity of the knapsack and the size of an item are vectors, and the process ends once any component of the vector capacity is exceeded. We observe that the two models differ slightly since in the knapsack hiring problem, both the value and size of an item (candidate) is stochastic.

We first give a reduction from the knapsack hiring problem to the multidimensional stochastic knapsack problem. For simplicity, we assume that the budget $B = 1$ without loss of generality. We construct an instance of 2-dimensional stochastic knapsack as follows - the knapsack capacity is $[1 ~ t]^\top$ where the first dimension represents the budget constraint  and the second dimension represents the number of allowed probes. The size of item $i$ is represented by the vector $[s_i ~ 1]^\top$ if the item exists when it is probed (happens with probability $p_i$) and $[0 ~ 1]^\top$ otherwise. The value $v_i'$ of item $i$ is set to the expected value obtained from candidate $i$, i.e., $v_i' = p_i v_i$.
With this reduction, the optimal solution to the knapsack hiring problem remains
unchanged if items deterministically contribute value $v_i' = p_i v_i$, as we
show in Appendix~\ref{app:equivalence}.

\citet{dean2005adaptivity} give a general $1+6d$-approximation to multidimensional stochastic knapsack, where $d$ is the number of knapsack constraints.
Directly applying this, we would get a 13-approximation in our 2-dimensional case. However, by leveraging the structure of the finite-probe problem, we can tighten this to a 10-approximation.

Without loss of generality, we assume that $s_i \le 1$ for all $i$ (otherwise the item would never fit in the knapsack).
We also normalize the number of probes to 1, so each item uses $1/t$ probes.
Let $\mu(i)$ denote the vector of expected size of item $i$, meaning $\mu(i) = [p_i s_i ~ ~ 1/t]^{\top}$. Let $\mu(S) = \sum_{i \in S} \mu(i)$. We use the notation $\mu_1(S)$ and $\mu_2(S)$ to denote the first and second components of $\mu(S)$ respectively. Further, let $\size(i)$ denote the vector of the realized size of item $i$ and let $\size(S) = \sum_{i \in S} \size(i)$.

\begin{algorithm}
\begin{algorithmic}[1]
\caption{KnapsackFiniteProbes$(p, v, s)$}
\label{alg:greedy_knapsack}
\STATE $m_1 \gets \max_i v_i' = \max_i p_i v_i$
\STATE $\mathcal{L} \gets$ the sequence of all items with $\|\mu(i)\|_1 \le 1/3$, sorted in non-increasing order of $\dfrac{v_i'}{\|\mu(i)\|_1}$
\STATE $m_{\mathcal{L}} \gets \sum_{i=1}^\ell v_i' (1 - \sum_{j \le i} \mu_1(j))$, where $\ell$ is the smallest integer such that $\sum_{i=1}^\ell \|\mu(i)\|_1 < 1$
\IF {$m_1 \ge m_{\mathcal{L}}$}
\STATE Probe the item with highest expected value $v_i'$
\ELSE
\STATE Probe the items in $\mathcal{L}$ until the knapsack is full
\ENDIF
\end{algorithmic}
\end{algorithm}

Our algorithm (Algorithm~\ref{alg:greedy_knapsack}) takes the better of two strategies: probing the item with highest expected value and probing a sequence of ``small'' items.
Exactly evaluating $m_{\mathcal{L}}^*$, the expected value of the second of these strategies, may be difficult; however, we can show that $m_{\mathcal{L}} = \sum_{i=1}^\ell v_i' (1 - \sum_{j \le i} \mu_1(j)) \le m_{\mathcal{L}}^*$.
We obtain $v_i'$ for item $i$ if and only if the first $i$ items in $\mathcal{L}$ all fit inside the knapsack.
Thus, $m_\mathcal{L}^* = \sum_{i=1}^\ell v_i' Pr[\|\size(\mathcal{L}_i)\|_\infty \leq 1]$, where $\mathcal{L}_i$ denotes the set of first $i$ items in $\mathcal{L}$.
By Claim \ref{clm:size}, $m_\mathcal{L}^* \geq m_{\mathcal{L}}$.
Note that Claim \ref{clm:size} applies since the constraint $\sum_{i=1}^\ell \|\mu(i)\|_1 < 1$ implies $\ell < t$.

\begin{claim}
\label{clm:size}
For any set $A$ of at most $t$ items, $Pr[\|\size(A)\|_\infty \leq 1] \ge 1 - \mu_1(A)$.
\end{claim}
See Appendix~\ref{app:proofs} for a proof.

Let $\greedy = \max\{m_1, m_\mathcal{L}\}$ be a lower bound on the expected value of Algorithm~\ref{alg:greedy_knapsack}.
Let $A$ be the random set of items that are probed by the optimal adaptive algorithm, and let $\opt$ be the expected value of the optimal algorithm.
\begin{lemma}[\cite{dean2005adaptivity}, Lemma 4.2 and Lemma 4.3]
\label{lem:deanknapsack}
$\opt \le (1 + 3\E{\|\mu(A)\|_1}) \greedy$.
\end{lemma}
For any adaptive algorithm, we can bound the expected size of the set of items probed using Lemma 2 from \citet{dean2004approximating}. In particular, we can bound the expected size of the first component as $\E{\mu_1(A)} \leq 2$. On the other hand, since the optimal adaptive algorithm can never probe more than $t$ items, $\E{\mu_2(A)} \leq 1$. Substituting these bounds into Lemma \ref{lem:deanknapsack} gives us the desired $10$-approximation:
\begin{align*}
  \opt &\le (1 + 3\E{\|\mu(A)\|_1}) \greedy \\
  &\le (1 + 3(\E{\mu_1(A) + \mu_2(A)})) \greedy\\
  &\le (1 + 3 \cdot 3) \greedy = 10 \cdot \greedy.
\end{align*}

\section{A Lower Bound for Stochastic Matching}
\label{sec:matching}
The hiring with uncertainty problem with $k=1$  can be viewed as a special case of the stochastic matching problem, which is as follows: given a graph $G = (V, E)$, probabilities $p_e$ and values $v_e$ for all $e \in E$, and patience parameters $t_v$ for all $v \in V$, the goal is to obtain a matching with maximum expected weight. As in our hiring problem, edges can be probed sequentially. If an edge $e$ is found to exist, it must be added to the matching, contributing value $v_e$. Each probe decreases the patience parameters of the incident vertices by 1, and when a vertex runs out of patience, it cannot be matched.

The state-of-the-art approach for this problem is to form a probing strategy by solving the linear program relaxation:
\begin{align*}
    \max_{x \in [0,1]^{|E|}} \sum_{e \in |E|} p_e v_e x_e && \text{s.t.} ~~ & \forall v ~ \sum_{e \in \delta(v)} x_e \le t_v \numberthis \label{eq:matching_lp} \\
    &&& \forall v ~ \sum_{e \in \delta(v)} p_e x_e \le 1
\end{align*}

This LP relaxation has been the primary approach for stochastic matching since~\citet{bansal2010lp}, yielding a 2.845-approximation for bipartite graphs~\cite{adamczyk2015improved} and a 3.224-approximation for general graphs~\cite{baveja2018improved}. However, little is known about the tightness of upper bound produced by the LP. Not only is there an integrality gap, but there is also a \textit{probing} gap -- the LP does not fully account for the random realizations of probes.

With $k=1$, the hiring problem is a special case of stochastic matching, since
it can be expressed as matching on a star-shaped graph. Thus, we can use it
provide a lower bound on the worst-case slack created by the LP. In
Appendix~\ref{app:lower_bound}, we provide an example showing that the gap between the LP value and the expected value of the optimal probing strategy must be at least $1 - 1/e$, meaning no probing strategy can approximate the optimal LP value to a factor better than $\frac{e}{e-1} \approx 1.581$.

\begin{figure*}%
\centering
\subfigure[Negative Correlation]{%
\label{fig:neg}%
\includegraphics[height=1.5in]{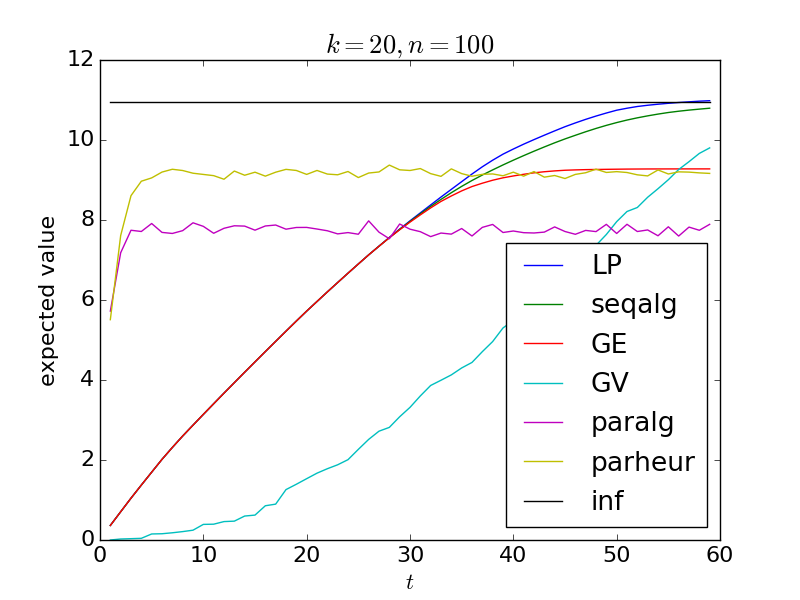}}%
\qquad
\subfigure[Positive Correlation]{%
\label{fig:pos}%
\includegraphics[height=1.5in]{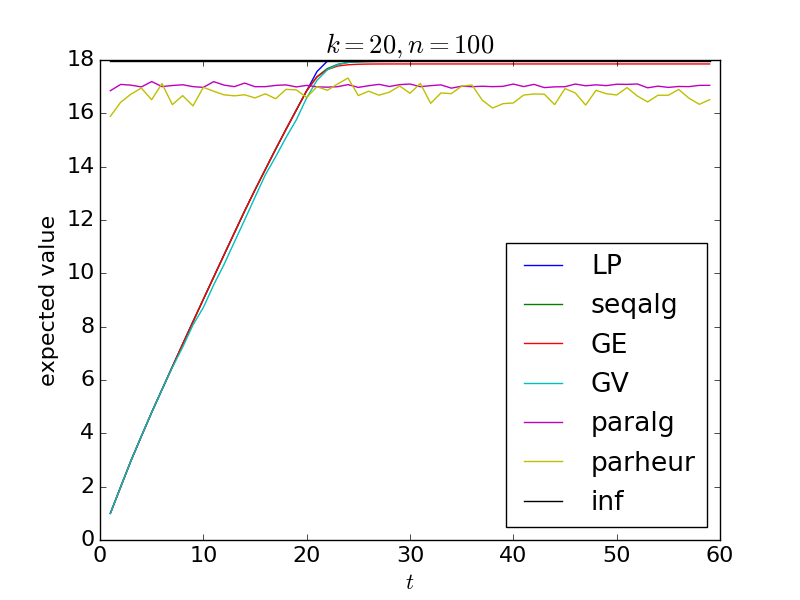}}%
\qquad
\subfigure[No Correlation]{%
\label{fig:no}%
\includegraphics[height=1.5in]{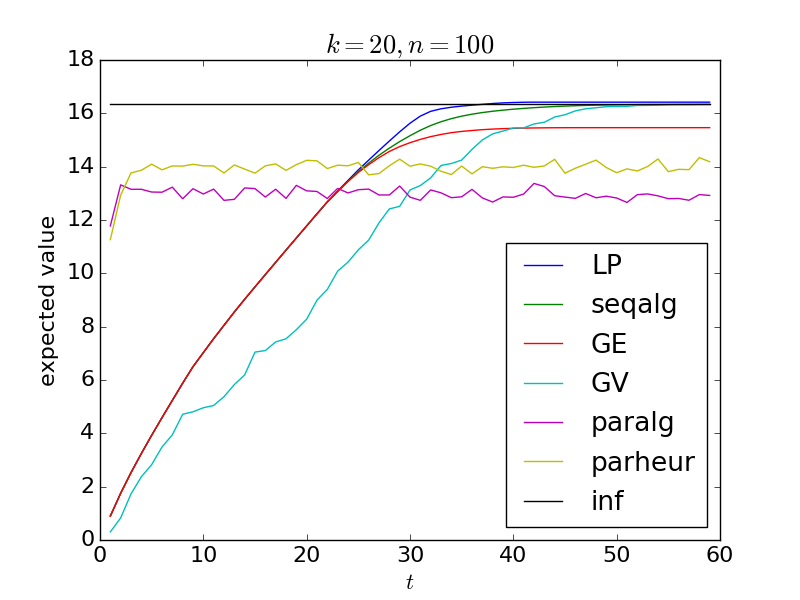}}%
\caption{Comparison of different algorithms on three simulated data sets.}
\end{figure*}

\section{Experiments}
\label{sec:experiments}

We test the performance of our algorithms for the \prob\ problem in both the sequential and parallel offers setting via simulations. We generate simulated data sets as follows. The values for $n = 100$ candidates are chosen uniformly at random from $[0, 1]$. We consider three models to generate the probabilities:
\begin{itemize}[topsep=0pt, partopsep=0pt, itemsep=-.5ex]
    \item \textbf{Negative correlation:} Higher-value candidates are less likely to accept offers. We sample $p_i$'s according to a Beta distribution, with $p_i \sim \text{Beta}(10(1-v_i), 10 v_i)$.
    \item \textbf{Positive correlation:} Higher-value candidates are more likely to accept: $p_i \sim \text{Beta}(10v_i, 10(1-v_i))$.
    \item \textbf{No correlation:} $p_i \sim \text{Uniform}[0,1]$.
\end{itemize}

On each of these data sets, we consider the performance of our three algorithms each with $k = 20$, namely
\begin{itemize}[topsep=0pt, partopsep=0pt, itemsep=-.5ex]
    \item {$\boldsymbol{\seqalg}$:} The dynamic programming algorithm from Section \ref{sec:star} to make $t$ sequential offers.
    \item {$\boldsymbol{\paralg}$:} The parallel approximation algorithm from Section \ref{sec:parallel}. We take the best solution over 100 random samples (of paths).
    \item {$\boldsymbol{\mathsf{parheur}}$: } We consider the following heuristic strategy to make offers in the parallel model. Note that the parallel approximation algorithm effectively partitions the set of candidates into up to $2k$ sets, selects the best $k$ of them, and makes offers to candidates in those sets in decreasing order of value. Our heuristic, then, is to randomly partition the set of candidates into $k$ disjoint sets and use the optimal single-slot solution from Section~\ref{sec:singlecandidate} on each set independently to decide which of them to make offers to. These offers can be made in parallel since the sets are disjoint.
\end{itemize}

For comparison, we include two natural greedy baselines.
We probe candidates in decreasing order of expected value $p_i \cdot v_i$ ($\boldsymbol{\mathsf{GE}}$) and value $v_i$ ($\boldsymbol{\mathsf{GV}}$).
We also plot two upper bounds on the value obtained by an optimal algorithm: $\boldsymbol{\mathsf{LP}}$, the value obtained by a natural LP relaxation (similar to \eqref{eq:matching_lp}), and $\boldsymbol{\mathsf{inf}}$, the optimal algorithm with $t = \infty$ (sort the candidates by decreasing value and make offers until $k$ candidates accept).

Figures~\ref{fig:neg},~\ref{fig:pos}, and~\ref{fig:no} demonstrate the performance of our algorithms on the three data sets with negative correlation, positive correlation, and no correlation respectively.
Beyond the theoretical guarantees, $\seqalg$ performs well empirically and dominates the greedy baselines, especially in the more natural setting where values and probabilities are negatively correlated.
$\seqalg$ is in general quite close to the LP upper bound -- much closer than the theoretical guarantee of 2. Thus, the LP is a fairly tight upper bound on the maximum value achievable by probing.

Even for moderately small values of $t$, $\seqalg$ outperforms $\paralg$, despite the fact that it makes 1 offer per time step when $\paralg$ makes multiple offers at a time.
Moreover, $\parheur$ almost always outperforms $\paralg$.
The relatively poor performance of $\paralg$ is to be expected.
Recall that $\paralg$ takes the solution tree to $\seqalg$ with $kt$ offers and probes candidates on the segments of a random path down this tree.
By construction, candidates on this path are sorted by value, so high-value candidates are concentrated in a small number of segments.
$\paralg$ can only select at most one candidate per segment, so it must ignore some high-value candidates.
In contrast, $\parheur$ partitions the candidates randomly, making it more likely that each set in the partition contains high-value candidates.

\subsection{Knapsack setting}
We also provide simulated results for a slightly modified version of the 10-approximation algorithm ($\boldsymbol{\mathsf{approx}}$) in the knapsack setting, where instead of calculating the lower bound $m_{\mathcal{L}}$ on the greedy strategy as in Algorithm~\ref{alg:greedy_knapsack}, we estimate the true expected value $m_{\mathcal{L}}^*$ by simulating runs of the greedy branch of the algorithm.
We compare the LP relaxation ($\boldsymbol{\mathsf{LP}}$) to our algorithm.
In addition, we compare to a natural greedy baseline ($\boldsymbol{\mathsf{greedy}}$), which probes items in decreasing order of $p_i v_i/s_i$.
We sample $s_i$ from a truncated Pareto distribution on $[0, 1]$, and sample $v_i \in [0, 1]$ from a Beta distribution positively correlated with $\sqrt{s_i}$.
We use $\sqrt{s_i}$ so that expected $v_i$'s exhibit diminishing returns in $s_i$.
We choose $p_i \sim \text{Uniform}[0, 1]$ and set a budget of $B = 1$.

\begin{figure}
    \centering
    \includegraphics[width=.35\textwidth]{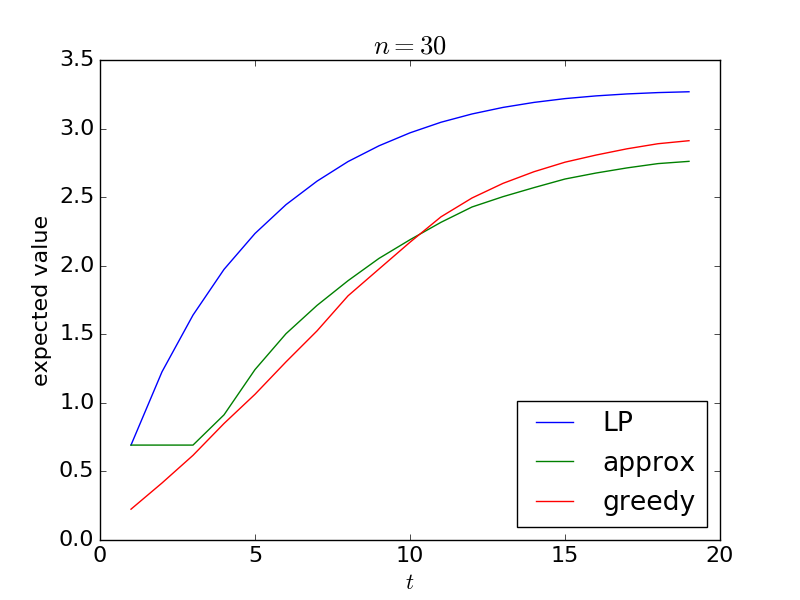}
    \caption{Experimental results for knapsack setting}
    \label{fig:knapsack_results}
\end{figure}

As the results in Figure~\ref{fig:knapsack_results} show, our algorithm performs roughly as well as $\boldsymbol{\mathsf{greedy}}$, but it does better for very small values of $t$.

\section{Conclusions}
\label{sec:conclusion}

With the increased use of data-driven techniques in hiring, predictions for employment outcomes are becoming increasingly accurate. Leveraging these predictions can be non-trivial, leading to the family of stochastic optimization problems we have considered here. As we have shown, imposing a finite number of offers can lead to highly complex solutions; however, by imposing an intuitive structure on the solution space, we are able to derive approximation algorithms that perform well, both theoretically and in practice.

\bibliography{bibfile}
\bibliographystyle{icml2019}
\clearpage
\appendix
\section{Appendix}
\label{app:gap}
\subsection{Gap between optimal adaptive and value-ordered strategies}
The following example shows a gap between the optimal adaptive strategy and the any value-ordered strategy in the sequential setting.
\begin{align*}
    (p_1, v_1) &= (1, 1) &
    (p_2, v_2) &= (q, 1) \\
    (p_3, v_3) &= (q, 1) &
    (p_4, v_4) &= (q(1-q)/(v-q),v)
\end{align*}
Here, we set $q = 0.63667$, and take the limit as $v$ goes to $\infty$.

The optimal value-ordered strategy is make offers to $1$, $2$, and possibly $3$ if $2$ rejects. This yields value $1+2q-q^2$. The optimal strategy is shown in Figure~\ref{fig:ex_sol} and yields value $1+2q-q^2 + (1-q)q^2(v-1)/(v-q)$. As $v \to \infty$, the
approximation ratio approaches
\[
\frac{1+2q-q^3}{1+2q-q^2} \approx 1.0788
\]

Moreover, this example demonstrates that simple greedy algorithms are suboptimal -- in particular, making offers greedily by decreasing $p_i$, $v_i$, and $p_i v_i$ all yield suboptimal value.

\subsection{Lower bound for LP-based stochastic matching}
\label{app:lower_bound}

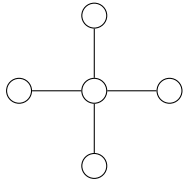
\begin{figure}[ht]
    \centering
    \begin{tikzpicture} 
    \node[circle, draw] (a) at (0, 0) {};
    \node[circle, draw] (b) at (1, 0) {};
    \node[circle, draw] (c) at (0, 1) {};
    \node[circle, draw] (d) at (0, -1) {};
    \node[circle, draw] (e) at (-1, 0) {};
    \draw[-] (a) -- (b);
    \draw[-] (a) -- (c);
    \draw[-] (a) -- (d);
    \draw[-] (a) -- (e);
    \end{tikzpicture} 
    \caption{Lower bound for integrality gap}
    \label{fig:int_gap}
\end{figure}

Consider the star graph as shown in Figure~\ref{fig:int_gap} with $n+1$ vertices, with $n$ leaves and 1 vertex in the middle. Each edge has $p_e = \frac{1}{n}$ and value 1. Let the number of probes be $t = n$. The value of the LP~\eqref{eq:matching_lp} is 1, assigning $x_e = 1$ to all edges. Since all edges are identical, any strategy is an optimal probing strategy, yielding expected value
\begin{align*}
    \sum_{i=1}^n \p{\frac{1}{n}} \p{1 - \frac{1}{n}}^{i-1}
    &= \frac{1}{n} \sum_{i=1}^n \p{1 - \frac{1}{n}}^{i-1} \\
    &= \frac{1}{n} \cdot \frac{1 - \p{1-\frac{1}{n}}^n}{\frac{1}{n}} \\
    &= 1 - \p{1-\frac{1}{n}}^n
\end{align*}
In the limit, this is $1 - 1/e$, so no probing strategy can be better than an $\frac{e}{e-1} \approx 1.581$-approximation.

\subsection{Deferred Proofs}
\label{app:proofs}

\begin{proof}[Proof of Claim~\ref{clm:size}]
For any set $A$ of at most $t$ items, $\size_2(A) = \sum_{i \in A} \size_2(i) = |A|/t \leq 1$. Further,
by Markov's inequality, we have $Pr[\size_1(A) \ge 1] \leq \E{\min\{\size_1(A), 1\}} \leq  \E{\sum_{i \in A} \min\{s_i, 1\}} = \mu_1(A)$. Consequently, we have $Pr(\|\size(A)\|_\infty < 1) \geq 1 - \mu_1(A)$.
\end{proof}

\begin{claim}
$\E{v_J} \geq \dfrac{1}{2} \E{v_I}$
\end{claim}

\begin{proof}
Let $W$ be the random set of elements on this segment, up to and including $I$, and let $W_x \subseteq W$ be the subset of those elements with value at least $x$. For ease of notation, we define $q_i = 1 - p_i$. Then, we can write
\begin{equation}
\label{eq:evI}
    \E{v_I}
    = \int_0^\infty \Pr[v_I \ge x] \dx
    = \int_0^\infty \sum_{i \in W_x} p_i \prod_{j < i} q_j \dx.
\end{equation}
Let $A$ be the random set of ``active'' candidates who will accept an offer if they receive one. Then, we have
\begin{align*}
\E{v_J} &= \int_0^\infty \Pr[v_J \ge x] \dx \\
&= \int_0^\infty \Pr[A \cap W_x \ne \emptyset] \dx \\
&=  \int_0^\infty \sum_{i \in W_x} p_i \cdot \Pr[I \ge i] \cdot \Pr\b{\bigcap_{j < i, j \in W_x} j \notin A} \dx \\
&= \int_0^\infty \sum_{i \in W_x} p_i \p{\prod_{j < i} q_j} \p{\prod_{j < i, j \in W_x} q_j} \dx \\
&= \int_0^\infty \sum_{i \in W_x} p_i \p{\prod_{j < i, j \in W_x} q_j^2} \p{\prod_{j < i, j \notin W_x} q_j} \dx
\numberthis \label{eq:evJ}
\end{align*}

We can write~\eqref{eq:evI} as
\begin{equation}
    \E{\sum_{i \in W_x} p_i \p{\prod_{j<i, j \in W_x} q_j}
    \p{\prod_{j < i, j \notin W_x} \mathbf{1}_{q_j}}},
\end{equation}
where the expectation is taken over the indicators $\mathbf{1}_{q_j}$ for $j < i, j \notin W_x$. Similarly, we write~\eqref{eq:evJ} as
\begin{equation}
    \E{\sum_{i \in W_x} p_i \p{\prod_{j<i, j \in W_x} q_j^2}
    \p{\prod_{j < i, j \notin W_x} \mathbf{1}_{q_j}}},
\end{equation}
Conditioning on the realizations of these indicators, it is sufficient to show that
\begin{equation}
    \sum_{i \in W_x} p_i \p{\prod_{j<i, j \in W_x} q_j^2}
    \ge \frac{1}{2}
    \sum_{i \in W_x} p_i \p{\prod_{j<i, j \in W_x} q_j},
\end{equation}
which is true by Claim~\ref{claim:gupta}.
\end{proof}

\begin{claim}[\cite{gupta2017adaptivity}, Claim 3.4]
\label{claim:gupta}
For any ordered set A of probabilities $\{a_1, a_2, \dots, a_{|A|}\}$, let $b_j$ denote $1-a_j$ for $j \in [1, |A|]$. Then,
\[
\sum_i a_i \p{\prod_{j < i} b_j}^2 \ge \frac{1}{2} \sum_i a_i \prod_{j < i} b_i
\]
\end{claim}

\subsection{Equivalence to Stochastic Knapsack}
\label{app:equivalence}
We must show that the the optimal solution remains unchanged whether values are received stochastically or deterministically.

It is easy to verify that the vector item sizes and knapsack capacities capture the budget and deadline requirements of the knapsack hiring problem.
However, in the reduction, item $i$ deterministically yields a value of $p_i v_i$ instead of value $v_i$ when $i$ is active (happens with probability $p_i$) and value 0 otherwise.

To account for this, observe that the optimal item to probe next depends only on the subset of remaining items, the number of probes left, and the capacity of the knapsack -- the value accumulated thus far has no bearing on the next action. Let $\opt(S, t, b)$ be the optimal value achievable with items (candidates) $S$, number of probes $t$, and budget $b$ remaining. The optimal strategy is then given by an exponential sized dynamic program, with the following recurrence 
\begin{align*}
\opt(S, t, b) = \max_{i \in S} \Big\{\ &p_i(v_i + \opt(S \backslash \{i\}, t - 1, b - s_i)) \\
+& (1-p_i) \opt(S \backslash \{i\}, t - 1, b) \Big\}. \numberthis \label{eq:knap_opt}
\end{align*}
Assuming inductively that $\opt(S', t, b)$ is unchanged whether $i$ contributes value $v_i$ with probability $p_i$ or deterministic value $p_i v_i$ for all smaller sets $S'$, we see that~\eqref{eq:knap_opt} is optimized by the same $i$ in both cases. Thus, the optimal strategy is unchanged in the deterministic and random cases and our reduction is complete.

\end{document}